\documentclass{article}

\usepackage{microtype}
\usepackage{graphicx}
\usepackage{subcaption}
\usepackage{booktabs} 

\usepackage{xurl}
\usepackage{hyperref}

\DeclareUnicodeCharacter{03F5}{\ensuremath{\epsilon}}

\usepackage{iclr2026_conference,times}

\usepackage{amsmath}
\usepackage{amssymb}
\usepackage{mathtools}
\usepackage{amsthm}

\usepackage[capitalize,noabbrev]{cleveref}

\theoremstyle{plain}
\newtheorem{theorem}{Theorem}[section]

\theoremstyle{definition}
\newtheorem{definition}[theorem]{Definition}

\theoremstyle{remark}

\usepackage[textsize=tiny]{todonotes}

\title{$\fhaim$: \textbf{F}ully \textbf{H}omomorphic \textbf{AIM} for Private Synthetic Data Generation}

\iclrfinalcopy

\author{%
Mayank Kumar$^{\dagger}$
Qian Lou$^{\dagger}$
Paulo L. Barreto$^{\ddagger}$
Martine De Cock$^{\ddagger}$
Sikha Pentyala$^{\ddagger}$\\
$^{\dagger}$University of Central Florida \qquad
$^{\ddagger}$University of Washington Tacoma\\
\texttt{\{mayank.kumar,qian.lou\}@ucf.edu, \{pbarreto,mdecock,sikha\}@uw.edu}
}

\usepackage{natbib}
\usepackage{amsmath,amssymb,amsfonts}
\usepackage{algorithmic}
\usepackage{graphicx}
\usepackage{textcomp}
\usepackage{xcolor}
\def\BibTeX{{\rm B\kern-.05em{\sc i\kern-.025em b}\kern-.08em
    T\kern-.1667em\lower.7ex\hbox{E}\kern-.125emX}}

\usepackage{hyperref}
\usepackage{url}
\usepackage{icomma}
\usepackage{soul}
\usepackage{subcaption}
\usepackage{graphicx}
\usepackage{multirow}
\usepackage{multicol}
\usepackage{caption}
\usepackage{tabularx}
\usepackage{algorithm}
\usepackage{algorithmic}

\usepackage{amsmath}
\usepackage{amssymb}
\usepackage{mathtools}
\usepackage{amsthm}
\usepackage{wrapfig}
\usepackage{booktabs} 
\usepackage{enumitem}

\usepackage{pdfpages}

\usepackage{indentfirst}
\usepackage{xcolor}

\newcommand{\fhaim}{\texttt{{FHAIM}}}
\newcommand{\pSEL}{\ensuremath{\pi_{\mathsf{SELECT}}}}
\newcommand{\pMSR}{\ensuremath{\pi_{\mathsf{MEASURE}}}}

\newcommand{\pCOMP}{\ensuremath{\pi_{\mathsf{COMP}}}}
\newcommand{\pCOMB}{\ensuremath{\pi_{\mathsf{COMB}}}}

\newcommand*\circledopt[1]{\tikz[baseline=(char.base)]{
    \node[shape=circle, draw, inner sep=1.5pt, font=\scriptsize] (char) {#1};}}

\newcommand{\pOM}{\ensuremath{\pi_{\mathsf{1way}}}}
\newcommand{\pTM}{\ensuremath{\pi_{\mathsf{2way}}}}
\newcommand{\pERR}{\ensuremath{\pi_{\mathsf{ERR}}}}
\newcommand{\pGUMBEL}{\ensuremath{\pi_{\mathsf{GUMBEL}}}}
\newcommand{\pGAUSS}{\ensuremath{\pi_{\mathsf{GAUSS}}}}

\newcommand{\xmark}{$\textcolor{red}{\times}$}%

\DeclareMathOperator*{\argmax}{arg\,max}

\newenvironment{myprotocol}[1][htb]{%
    \floatname{algorithm}{Protocol}
   \begin{algorithm}[htb]%
   \footnotesize{}
  }{\end{algorithm}}

\begin{document}

\raggedbottom

\maketitle

\begin{abstract}
 Data is the lifeblood of AI, yet much of the most valuable data remains locked in silos due to privacy and regulations. As a result, AI remains heavily underutilized in many of the most important domains, including healthcare, education, and finance. Synthetic data generation (SDG), i.e.~the generation of artificial data
  with a synthesizer trained on real data, offers an appealing solution to make data available while mitigating privacy concerns, however existing SDG-as-a-service workflow require data holders to trust providers with access to private data.
 We propose $\fhaim$, the first fully homomorphic encryption (FHE) framework for training a marginal-based synthetic data generator on encrypted tabular data. $\fhaim$ adapts the widely used AIM algorithm to the FHE setting using novel FHE protocols, ensuring that the private data remains encrypted throughout and is released only with differential privacy guarantees. Our empirical analysis show that $\fhaim$ preserves the performance of AIM while maintaining feasible runtimes.
\end{abstract}

%
%
\section{Introduction}\label{sec:intro}
Despite the fact that modern AI applications are heavily data-driven, much of the most valuable real-world data remains siloed within research centers, hospitals, companies, and financial institutions. These silos are typically protected by strict access controls and regulations due to the sensitive and personal nature of the data they contain. This,  while necessary, significantly hampers data sharing and reuse. As a result, AI remains underutilized in many high-impact domains, including healthcare, genomics, biomedicine, commerce, education, and finance.

\textit{Synthetic data generation (SDG)}, i.e., the process of generating artificial data using a synthesizer trained on real data, offers an appealing approach to facilitate data sharing while mitigating privacy concerns \cite{hu2023sok}. When done well, synthetic data has the similar characteristics
as the original data but, crucially, without replicating personal information. 
This has led to SDG becoming a widely adopted tool for enabling data access \cite{pets_winner, V-CHAMPS}.
This success has facilitated numerous SDG-as-a-service 
platforms\footnote{Examples include Mostly AI, Gretel (acquired by NVIDIA), Syntegra, MDClone, Hazy, and Tonic.} where data holders outsource the process of SDG to 
third parties.

While many SDG-as-a-service platforms are marketed as privacy-preserving, they typically provide at most only \textit{output privacy}, most commonly through techniques such as differential privacy (DP) \cite{dwork2006calibrating} which aim to ensure that the released synthetic data does not leak information about individual records in the real data that was used to train the generator.\footnote{DP substantially improves the robustness of synthetic data against adversarial attacks \cite{golob2025SatML}.} However, these approaches implicitly assume that the service provider (the third party) has full access to the raw training data in plaintext. In many practical settings, data holders are subject to strict privacy regulations
and organizational constraints that prohibit disclosure of sensitive data to third parties, as the risks associated with personal data exposure can be harmful \cite{Martin2017,nytimes2019location}.  \textit{Input privacy}, i.e.~protecting the confidentiality of the input data, remains an open problem in outsourced synthetic data generation. 

\begin{table*}[ht!]
    \centering
    \small
    \setlength{\tabcolsep}{3pt}
    \renewcommand{\arraystretch}{1.05}
    \begin{tabularx}{\linewidth}{@{}l c >{\raggedright\arraybackslash}X c c@{}}
    \toprule
         \textbf{Approach}  & \textbf{Input Privacy} & \textbf{Output Privacy} & \#\textbf{Data Holders} & \#\textbf{Computing Parties} \\
    \midrule
         AIM \cite{mckenna2022aim} & \xmark & Global DP on unencrypted data & 1  & 1 \\
         CAPS \cite{pentyala2024caps} & MPC & Global DP in MPC protocols & 1 or Multiple & Multiple$^{*}$ \\
         FLAIM \cite{maddock2024flaim} & FL  & Distributed DP on unencrypted data & Multiple & Multiple \\
         \midrule
         $\fhaim$ (ours) & FHE & Global DP in FHE protocols & 1 & 1 \\
         \bottomrule
    \end{tabularx}
    \caption{\textbf{Input Privacy Extensions for AIM.} $^{*}$: needs more than 1 non-colluding computing server, efficient with 3.}
    \label{tab:related}
\end{table*}

While approaches based on secure multiparty computation \cite{pentyala2024caps} and federated learning \cite{maddock2024flaim} offer forms of input privacy, they require co-ordination among multiple non-colluding computational parties, making them unsuitable for SDG-as-a-service deployed \textit{by a single service provider for a single data holder}. 
To our knowledge, no prior work addresses providing input privacy in such settings.

We propose to leverage fully homomorphic encryption (FHE) ~\cite{gentry2009fullyStanford, rivest1978data} to enable synthetic data generation in outsourced settings, presenting the first FHE-based SDG of its kind. Our approach allows an SDG-as-a-service provider to train synthetic data generators directly on encrypted data without ever accessing the raw data (data holders encrypt their data before providing it to the service provider), thus providing input privacy. We focus on tabular data, a widely used data modality in privacy-sensitive domains such as healthcare, finance, and public administration.

To demonstrate the feasibility of our approach, we propose $\fhaim$, built on AIM \cite{mckenna2022aim}, a state-of-the-art method for DP tabular data generation  \cite{chen2025benchmarking, tao2021benchmarking}. AIM is a marginal-based method that fits a joint probability distribution on the real data and subsequently samples from it to generate synthetic data. AIM does this in an iterative manner: in each iteration (1) a subset of attributes is selected for which the marginal probability distribution over the real data differs the most from the current version of the synthetic data (\textit{select step}); (2) this marginal distribution is estimated over the real data (\textit{measure step}); and (3) this estimated marginal is added to a probabilistic graphical model and a new version of the synthetic data is generated (the \textit{generate step}). To provide formal output privacy with DP, the select step is non-deterministic, with attribute subsets that yield a higher mismatch more likely to be chosen; likewise, in the measure step, Gaussian noise is added to the marginal probability measurements. 


The key idea of $\fhaim$ is to perform all  data-dependent operations in the select and measure steps
within FHE. We design novel FHE protocols for these operations, effectively implementing DP within FHE (DP-in-FHE), meaning that the outputs of the FHE operations are protected with DP noise, even after decryption. 
Specifically, the data holder provides the service provider with an encrypted version of their private dataset, along with encrypted unit noise samples. This enables the service provider to perform the data-dependent operations in the select and measure steps on the encrypted data.
Once the noisy marginals are obtained, they are decrypted, and the provider proceeds with the generation step in-the-clear, since this step is data-independent and operates only on the already-privatized statistics. Importantly, the service provider never observes the noise used to provide DP guarantees. Hence, in addition to never seeing the real data, the service provider never observes the probability distributions estimated on the real data, thus providing stronger privacy guarantees. 
%
Our main contributions are: 
\begin{itemize}[noitemsep, topsep=0pt,parsep=0pt,partopsep=0pt, leftmargin=*]
    \item \textbf{First FHE-based SDG framework.} We introduce $\fhaim$, the first system that enables training marginal-based synthetic data generators directly on fully homomorphically encrypted data, providing input privacy without requiring multiple non-colluding parties.
    
    \item \textbf{Novel DP-in-FHE protocols.} We design efficient FHE protocols for marginal computation ($\pi_{\text{COMP}}$), differentially private query selection ($\pi_{\text{SELECT}}$), and noisy measurement ($\pi_{\text{MEASURE}}$), implementing the Gaussian and exponential mechanisms entirely within the encrypted domain. Towards this, we propose an efficient encrypted memory layout that enables marginal computation with multiplicative depth depending only on marginal degree $k$, ensuring scalability. We also propose to replace the $L_1$-norm quality score with a squared $L_2$-norm in the select step, avoiding unstable polynomial approximations of the absolute value function in FHE.
    
    
    
\end{itemize}

We demonstrate practical feasibility on three real-world datasets, achieving runtimes of $\sim$11 to $\sim$30 minutes while preserving the statistical utility and downstream ML performance of the original AIM algorithm.

%
%
\section{Related Work}\label{sec:related}
There is ample work on SDG with output privacy; see \cite{du2024systematicassessmenttabulardata} and references therein. The approach that we develop in this paper for AIM, can be extended to other SDG algorithms such as MST \cite{mckenna2021winning} and RAP \cite{vietri2022private}.

In contrast to the wealth of literature on SDG with formal output privacy guarantees, existing research on SDG with input privacy is very limited.
%
\cite{pentyala2024caps} leverage secure multiparty computation (MPC) to provide input privacy. While designed for scenarios with multiple data holders, e.g., multiple hospitals who want to train a generator over their combined data, an MPC setup can in principle be used for the single data holder scenario that we consider in this paper. Indeed, the data holder could send encrypted shares of their data to a set of MPC servers, who then proceed to train the generator. But MPC has a strong non-collusion assumption: if the computing servers collude, they can reconstruct the original data. As such, it may be difficult to set up a credible SDG-as-a-service based on MPC in the cloud of a single provider. Indeed, if all computing servers in an MPC setting reside with the same provider or entity (e.g.~a federal agency), then the risk and the perception that they may collude and reconstruct the original data may be high. 

Federated Learning (FL) based approaches, such as FLAIM \cite{maddock2024flaim}, provide input privacy in distributed, multi-silo settings by ensuring that raw data never leaves each data holder. This is achieved by letting each client perform the select and measure computations on site, and only share (privatized) statistics with a central server who updates a global probabilistic model that can be accessed by all the data holders. This approach, which is designed for multiple data holder scenarios, does not lend itself to the single service provider, single data holder scenario that we consider in this paper. Even in a hypothetical FL setup with only one data holder, that data holder would be performing a lot of the computations, going against the intent of outsourcing the SDG training to a service provider.  

Table \ref{tab:related} 
summarizes the distinction between our proposed approach and the closest existing work.
To the best of our knowledge, there is no work on training of synthetic data generators over data that is encrypted with fully homomorphic encryption (FHE), neither with AIM nor with any other SDG algorithm for any kind of data modality. Existing research 
on combining FHE and DP to compute DP statistics such as marginals \cite{roy2020crypte, ushiyama2022homomorphic, ushiyama2021construction, bakas2022private} 
focuses on individual statistics rather than SDG training.
For completeness, we mention that while 
recent work explored ``DP for free'' from HE noise \cite{ogilvie2024differential}, this faces significant barriers (data-dependent variance, noise growth, parameter-dependence) and remains largely theoretical. Our approach explicitly considers DP noise in FHE for standard DP guarantees and can be adapted to future progress in this area.

%
%
\section{Preliminaries}\label{sec:prelims}

\subsection{Fully Homomorphic Encryption}

Given a target plaintext transformation $f(\cdot)$ that maps an input $x$ to an output $y$, fully homomorphic encryption (FHE) defines a corresponding homomorphic function $g(\cdot)$ that operates on the encrypted input $\mathrm{Enc}_{pk}(x)$. The result of this operation remains in the ciphertext domain, satisfying the core consistency equation: $\mathrm{Dec}_{sk}(g(\mathrm{Enc}_{pk}(x))) = f(x)$ where $\langle pk,sk\rangle$ are a mathematically linked \textit{public key} and \textit{secret key} pair. The cryptographic security of FHE schemes relies on the Learning With Errors (LWE) assumption ~\cite{regev2009lattices}. Practical implementations frequently utilize the Ring LWE (RLWE) ~\cite{lyubashevsky2013ideal} variant, which is structured over the polynomial ring $R_Q = \mathbb{Z}_Q[x] / \langle x^N + 1 \rangle$. In this setting, the polynomial degree $N$ is a critical security parameter, typically selected within the range of $2^{12}$ to $2^{16}$. Note that in the following sections, $[\![x]\!]=\mathrm{Enc}_{pk}(x)$.

FHE facilitates computation on encrypted data through a suite of fundamental primitives. Specifically, homomorphic addition (\texttt{HE.Add}) and homomorphic multiplication (\texttt{HE.Mult}) execute arithmetic operations on pairs of ciphertexts, while homomorphic rotation (\texttt{HE.Rot}) performs cyclic shifts on encrypted vectors.
Correctness of FHE operations is guaranteed provided two conditions are met: (1) the ciphertext level remains within the predetermined multiplication depth limit, and (2) the accumulated noise does not exceed the decryption radius (typically $Q/2$ for the ciphertext modulus $Q$). Under these constraints, decryption maintains accuracy within a prescribed error tolerance, typically $10^{-9}$ to $10^{-12}$. 

In this paper, we employ the CKKS scheme \cite{cheon2017homomorphic}, which supports approximate arithmetic on real numbers. While the computation of marginals (i.e., counts) in AIM is integer-valued, subsequent operations such as the noise addition to provide DP guarantees involve real-valued functions. 
Though integer based schemes such as BFV~\cite{fan2012somewhat} and BGV~\cite{brakerski2014leveled} are efficient for marginals computations, they would require costly fixed-point encoding and rescaling operations for each real-valued computation. This makes CKKS the most efficient choice for our work with DP-in-FHE.
In Section \ref{sec:results} we compare the utility of synthetic data generated with $\fhaim$ (over encrypted data) with AIM (over plaintext data), empirically demonstrating that the use of an approximate 
FHE scheme has minimal impact on the results.

\subsection{Differential Privacy}
Differential Privacy (DP) is a formal privacy notion that bounds the impact an individual in a dataset $D$ can have on the output of an algorithm or operation $\mathcal{A}$ performed over $D$ \cite{dwork2006calibrating}. 
Formally, $\mathcal{A}$ is called $(\varepsilon,\delta)$-DP if for all pairs of neighboring datasets $D, D'$ $\in$ $\mathbb{D}$, i.e., $D'$ can be obtained
from $D$ by adding or removing a single record (record-level privacy), and for all subsets $O$ of $\mathcal{A}$'s range,
$\mbox{P}(\mathcal{A}(D) \in O) \leq e^{\varepsilon} \cdot \mbox{P}(\mathcal{A}(D') \in O) +\delta$. The parameter $\varepsilon \geq 0$ denotes the \textit{privacy budget} or privacy loss, while $\delta \geq 0$ denotes the probability of violation of privacy, with smaller values indicating stronger privacy guarantees in both cases.
$\mathcal{A}(D)$ and $\mathcal{A}(D')$ could e.g.~be marginal probability distributions estimated on datasets $D$ and $D'$ respectively.

A DP operation $\mathcal{A}$ is commonly created from an operation $f$ by adding noise that is inversely proportional to $\varepsilon$ and proportional to the \textit{sensitivity} of $f$, in which the sensitivity measures the maximum impact a change in the underlying dataset can have on the output of $f$. For example, let $f: \mathbb{D} \to \mathbb{R}^m$ and let $D \sim D'$ denote that $D$ and $D'$ are neighboring datasets, then the $L_2$ sensitivity of $f$ is $\Delta(f) = \max_{D \sim D'} \|f(D) - f(D')\|_2$, while the $L_1$ sensitivity of $f$ is based on the L1-norm, i.e., $\Delta(f) = \max_{D \sim D'} \|f(D) - f(D')\|_1$.

\begin{definition}[Gaussian Mechanism]
Let $f: \mathbb{D} \to \mathbb{R}^m$, then  the Gaussian Mechanism adds i.i.d.~Gaussian noise with scale proportional to $\Delta(f)$ to each entry of $f(D)$. Mathematically, $\mathcal{A}(D) = f(D) + \sigma\Delta(f)\mathcal{N}(0, \mathbf{I})$, where $\sigma$ is a scaling coefficient and  $\mathbf{I}$ is the $m \times m$ identity matrix.
\end{definition}

Similarly, the exponential mechanism is commonly used to make a private choice from a set $\mathcal{R}$ of possible outcomes.

\begin{definition}[Exponential Mechanism] Given a quality score function $s: (\mathcal{R},\mathbb{D}) \to \mathbb{R}$
and a privacy budget $\varepsilon \geq 0$, the exponential mechanism samples an output $r \in \mathcal{R}$ from the probability distribution:
$\Pr[\mathcal{A}(D) = r] \propto \exp\left(\frac{\varepsilon \cdot s(r,D)}{2\Delta}\right)$,
where $\Delta = \max_{r \in \mathcal{R}} \Delta(s(r,.))$ denotes the global sensitivity of the quality function.
\end{definition}
We employ the Gumbel-Max trick as a computationally efficient alternative to the standard exponential mechanism 
\cite{mcsherry2007mechanism}. For each candidate $r \in \mathcal{R}$, we compute $s(r,D) + G_r$ where $G_r \sim \text{Gumbel}(0, \beta)$ with scale $\beta = \frac{2\Delta}{\varepsilon}$, and return $\argmax_{r \in \mathcal{R}} (s(r,D) + G_r)$. This procedure is equivalent to the Report Noisy Max (RNM) algorithm and generates samples from the exact exponential mechanism distribution.

The \emph{post-processing property} of DP guarantees that if $\mathcal{A}$ is DP,
then $g \circ \mathcal{A}$ is also DP,
where $g$ is an arbitrary function. I.e., 
any arbitrary computations performed on DP output preserve DP without any effect on the privacy budget $\varepsilon$. 

\subsection{Synthetic Data Generation}\label{sec:sdgpremlim}
A private tabular dataset $D$ consists of $N$ instances, where each instance $x \in D$ is defined by a set of $d$ attributes. Each attribute $x_i$ takes values from a finite, discrete domain $\Omega_i$ of size $|\Omega_i| = \omega_i$. The complete domain of the data is the cartesian product of these individual domains, denoted by $\Omega = \prod_{i=1}^{d} \Omega_i$. Our objective is to generate a synthetic dataset $\widehat{D}$ that accurately reflects the statistical properties of $D$ while satisfying input and output privacy guarantees.

For a subset of attribute indices $w \subseteq \{1, \dots, d\}$, a marginal $q_w(D)$ is a vector in $\mathbb{R}^{|\Omega_w|}$ representing the frequency in $D$ of each 
element of 
the sub-domain $\Omega_w = \prod_{i \in w} \Omega_i$. Specifically, the entry corresponding to a configuration $y \in \Omega_w$ is given by: $q_w(D)_y = 
\sum_{x \in D} \mathbb{I}(x_w = y)$ where $\mathbb{I}(\cdot)$ is the indicator function and $x_w$ is the projection of $x$ onto the attributes in $w$. 
A $k$-way marginal refers to a marginal where the subset of attributes has cardinality $|w| = k$. 

To generate $\widehat{D}$, we use the marginal-based method AIM  
\cite{mckenna2022aim}. It accepts a workload $\mathcal{W} = \{w_1, w_2, \dots, w_m\}$, which is a pre-defined collection of attribute subsets of interest. Given a total privacy budget $\varepsilon$, the goal is to ensure $q_w(\widehat{D}) \approx q_w(D)$ for all $w \in \mathcal{W}$. To achieve this, AIM first computes noisy 1-way marginals
on $D$ to initialize the probabilistic graphical model (\textbf{initialization} step, which is a data-dependent operation) 
and then iteratively constructs the synthetic data distribution by repeating:

\begin{enumerate}[noitemsep, topsep=0pt,parsep=0pt,partopsep=0pt, leftmargin=*]
    \item \textit{\textbf{Select}}: 
    In each iteration $t$, AIM identifies a workload query $w^* \in \mathcal{W}$ that maximizes the error between $D$ and $\widehat{D}$ using the exponential mechanism. The quality score $s^t(w, D)$ for  $w$ 
    is given as: $s^t(w, D) = \alpha_w \cdot (\| q_w(D) - q_w(\widehat{D}) \|_p - \rho )$ where $\| \cdot \|_p$ denotes a norm (typically the $L_1$ or $L_2$ norm), $\alpha_w$ is a query-specific weight, and $\rho$ is a penalty term. To implement this efficiently, we employ the Gumbel-Max trick: $w^* = \arg\max_{w \in \mathcal{W}} \left( s^t(w, D) + G_w \right)$ where $G_w \sim \text{Gumbel}(0, \beta)$ with $\beta$ proportional to the sensitivity of $s^t(w, D)$. For details, please refer to Eq 1 in \cite{mckenna2022aim}.   
    This step is data-dependent, as the quality score is computed directly using the private database $D$. 
    \item \textit{\textbf{Measure}}: Once $w^*$ is selected, the algorithm computes a noisy measurement of the true marginal using the Gaussian mechanism $y_{w^*} = q_{w^*}(D) + \mathcal{N}(0, \sigma^2_t\mathbf{I})$ where the noise scale $\sigma_t$ for the $t^{th}$ iteration is determined by the remaining privacy budget (Note that sensitivity of $q$ is 1). This is a data-dependent step. 
    \item \textit{\textbf{Generate}}: AIM maintains a compact representation of the synthetic data, typically using a Probabilistic Graphical Model \cite{mckenna2019graphical}. After each measurement, the model is updated to maintain consistency with all previously measured marginals while minimizing the objective function that finds a distribution that best explains the noisy measurements $y_{w^*}$, effectively refining $\widehat{D}$. This step 
    depends only on the noisy measurements and not on the raw data. 
\end{enumerate}


\noindent
\textit{\textbf{Compute:}}
In practical implementations of AIM, it is common to first compute the exact values of all marginals in  
$ \mathcal{W}$ on $D$, including the 1-way marginals used for initialization and all higher-order marginals that may be selected during the adaptive iterations. We refer to this data-dependent preprocessing as the \textbf{compute} step.\footnote{In AIM, the precomputed exact answers are used by the select step to evaluate approximation error, while noisy versions are generated only when a marginal is to be  measured.}
This step does not modify the AIM algorithm or its privacy guarantees, since no true marginals are revealed.  We note that such precomputation is particularly efficient in FHE (see Prot. \ref{prot:comp_he}).





\textit{Privacy Guarantee} \cite{mckenna2022aim}: By using the Gaussian Mechanism for measurements and the Exponential Mechanism for selection,  AIM satisfies $(\varepsilon, \delta)$-DP with respect to $D$. Specifically, the privacy budget is managed using the composition properties of Zero-Concentrated Differential Privacy (zCDP), which are subsequently converted to $(\varepsilon, \delta)$. 






%
%
\section{$\fhaim$: System Description}\label{sec:method}
Let \textit{Alice} be a data holder with a  private dataset $D$ and let Bob be a SDG-as-a-service provider who owns an SDG process $\mathcal{G}$. Alice wishes to utilize Bob’s service to generate a synthetic dataset $\widehat{D}$ without sharing her raw private  data $D$. We propose a system comprising four entities (see Fig.~\ref{fig:entities}), following a model similar to \cite{ushiyama2021construction}. 
\begin{itemize}[noitemsep, leftmargin=*,topsep=0pt]
    \item \textbf{\textit{Data Entity (DE)}}: represents the data holder (\textit{Alice}). The DE encrypts the private dataset $D$ and other auxiliary information (such as generated unit noise samples) using the public key $pk$ received from the Crypto-Service Entity and sends the encrypted data to the Computation Entity. After this initial upload, the DE is not involved in further computations.
    
    \item \textbf{\textit{Computation Entity (CE)}}: provides the computational infrastructure (e.g., a cloud platform like Amazon SageMaker) to run  $\mathcal{G}$ including the FHE-based SDG 
    protocols (see Sec.~\ref{sec:prot}). The CE only operates on data that is either encrypted or protected by DP guarantees.
    
    \item \textbf{\textit{Generation Entity (GE)}}: represents the service provider (\textit{Bob}) who owns the 
    trademark SDG algorithm $\mathcal{G}$. The GE uploads the FHE-compatible protocols and the iterative SDG logic to the CE. GE may own CE.
    
    \item \textbf{\textit{Crypto-Service Entity (CSE)}}: This entity manages the FHE keys. It generates the key pairs $\langle pk,sk \rangle$, provides the public encryption key $pk$ to the DE, and performs the decryption of the noisy DP statistics requested by the CE using $sk$. Any data seen by the CSE is already protected by DP noise added within the encrypted domain. \textit{Alice} or any other trusted third-party (such as AWS Cryptography) can be a CSE. 
\end{itemize}

The above proposed system works for any SDG algorithm. In particular, we adapt the system in Fig.~1 to AIM \cite{mckenna2022aim} and develop novel FHE-based SDG protocols to this end. 

\begin{figure*}[t!]
    \centering
    \begin{subfigure}[t]{0.49\textwidth}
        \centering
        \includegraphics[width=\linewidth]{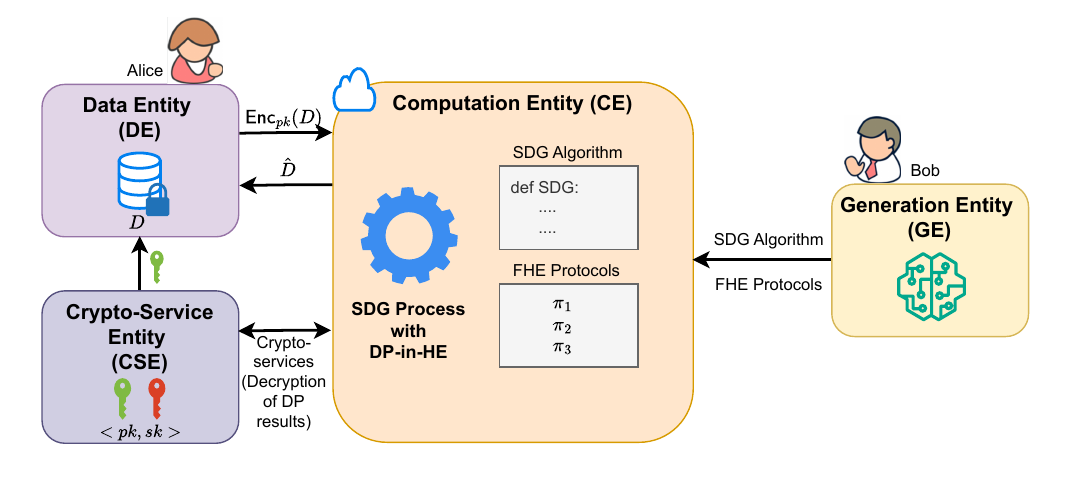}
        \caption{System entities (see Sec.~\ref{sec:method})}
        \label{fig:entities}
    \end{subfigure}
    \hfill
    \begin{subfigure}[t]{0.49\textwidth}
        \centering
        \includegraphics[width=\linewidth]{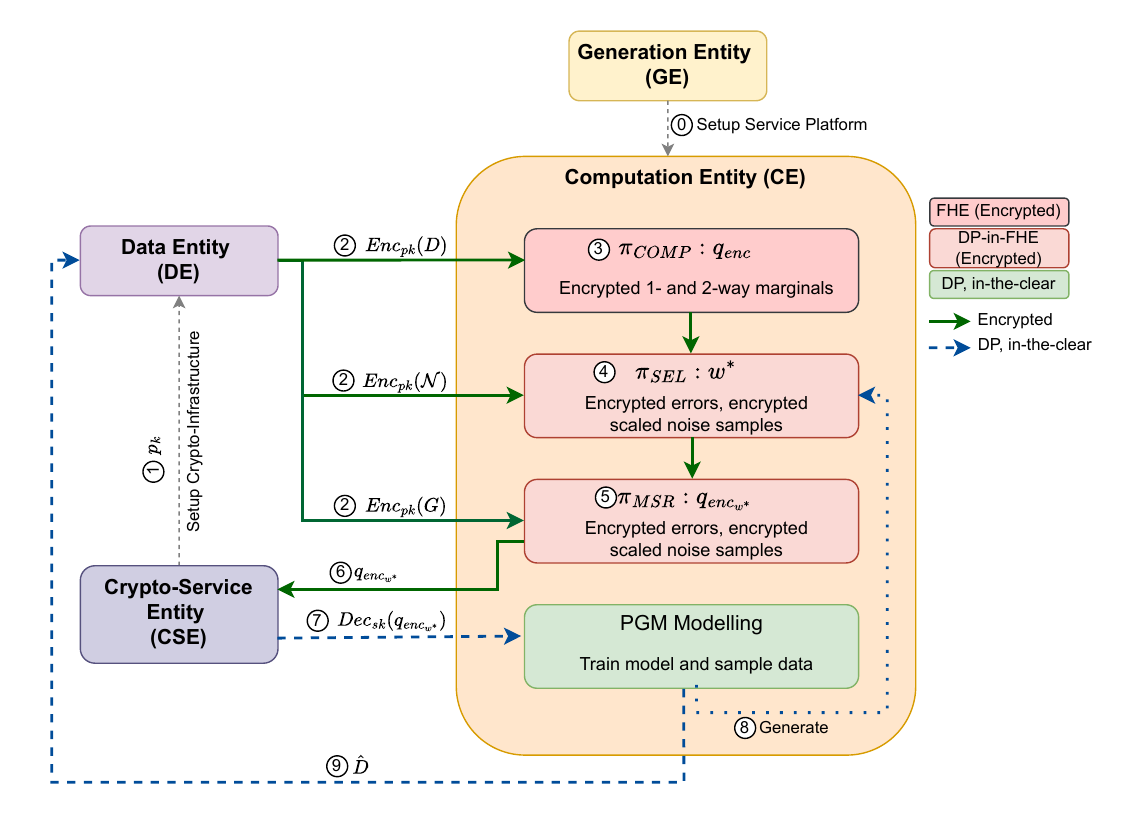}
        \caption{$\fhaim$ workflow (see Sec.~\ref{sec:method})}
        \label{fig:process}
    \end{subfigure}
    \caption{\textbf{Overview of $\fhaim$.} (Left) system entities; (right) end-to-end workflow.}
    \label{fig:overview}
\end{figure*}

\paragraph{$\fhaim$ Workflow.} The proposed $\fhaim$ workflow, illustrated in Fig.~\ref{fig:process}, begins with setting up the service platform and crypto-infrastructure. The CSE generates an FHE key pair and provides the public encryption key $pk$ to the DE (Step \circledopt{1}). The DE then preprocesses and encrypts their private dataset $D$, along with a collection of pre-generated unit noise samples (see Sec.~\ref{app:noise}), and uploads the resulting ciphertexts $Enc_{pk}(D)$, $Enc_{pk}(\mathcal{N})$, and $Enc_{pk}(G)$ to the CE (Step \circledopt{2}). Using the FHE-compatible protocols and iterative logic provided by the GE (in Step \circledopt{0}), the CE executes the DP-in-FHE synthetic data generation process. The CE first computes encrypted 1- and 2-way marginals ($q_{enc}$) by executing the FHE protocol $\pCOMP$ on the encrypted data (Step \circledopt{3}). It then iteratively follows the AIM workflow by executing the FHE protocols $\pSEL$ (Step \circledopt{4}) and $\pMSR$ (Step \circledopt{5}) to obtain an encrypted differentially private marginal $q_{enc_{w^*}}$ for $w^*$. The noisy encrypted marginal is sent to the CSE for decryption using the secret key $sk$ (Step \circledopt{6} \& \circledopt{7}). Since DP noise is added within the encrypted domain, the CSE can safely decrypt and return the plaintext noisy marginals to the CE. The CE then executes the modeling in the ``generate'' step of AIM in the clear (Step \circledopt{8}), leveraging the post-processing property of differential privacy to refine the synthetic data distribution without further compromising input privacy. After a fixed number of iterations, the trained DP model finally generates $\widehat{D}$ and sends it to DE (Step \circledopt{9}).

\vspace{0.3em}
\noindent
\textbf{Trust Model.} We consider CE and CSE to be semi-honest, i.e., they follow the protocols but attempt to learn about $D$. DE does not trust GE, and GE trusts CE with its trademark algorithm. We additionaly consider that CE and CSE do not collude. The keys are secured and are held only by CSE for decryption and DE for encryption. The system is designed to provide DP guarantees to the individuals in $D$.

\vspace{0.3em}
\noindent
\textbf{Noise Generation.}A central challenge in DP-in-FHE is generating the randomness required to sample DP noise. We adopt a pre-sampling strategy where DE encrypts and uploads a sufficient collection of unit noise samples during the initial upload phase to support the full AIM execution. While this adds a one-time computational burden on the DE, we discuss the challenges and propose strategy in App.\ref{app:noise}.

\section{$\fhaim$: FHE-SDG Protocols}\label{sec:prot}

\noindent
\textbf{FHE Protocol for Compute ($\pCOMP$).}
To enable efficient arithmetic, $D$ is first one-hot encoded (OHE), transforming categorical attributes into binary vectors. 
This encoding allows the CE to compute marginal histograms using native additions and multiplications, thereby avoiding costly homomorphic comparison operations. Protocol \ref{prot:comp_he} executes the compute step in FHE. It takes the column-wise encrypted one hot encoded dataset $[\![D]\!]$ to output encrypted 1- and 2-way marginals $[\![q]\!]$ (as is considered by default in AIM implementation).

\begin{myprotocol}
   \caption{$\pCOMP$: FHE Protocol for \textbf{COMPUTE}}
   \label{prot:comp_he}
    \textbf{Input:} Encrypted one-hot encoded data $[\![D]\!] \in \{0,1\}^{N \times \sum_{i=1}^{d}\omega_i}$, 
   workload  $\mathcal{W}$ with attribute subsets from $D$ \\
    \textbf{Output:} $[\![q]\!]$, $\forall w \in \mathcal{W}$
   \begin{algorithmic}[1]
   \FORALL{$w \in \mathcal{W}$}
       \IF{$|w| = 1$} 
           \FOR{$j \leftarrow 1$ \bf{to} $\omega_w$}
               \STATE $[\![q_w[j]]\!] \leftarrow \textsc{Combine}\!\left(\sum_{i=1}^{N} [\![D_w[j]]\!]\right)$
           \ENDFOR
       \ELSIF{$|w| = 2$, i.e., $w = \{a_1, a_2\}$}
           \FOR{$j \leftarrow 1$ to $\omega_{a_1}$}
                \FOR{$k \leftarrow 1$ to  $\omega_{a_2}$}
               \STATE $[\![v]\!] \leftarrow \sum_{i=1}^{N} [\![D_{a_1}[i,j]]\!] \cdot [\![D_{a_2}[i,k]]\!]$
               \STATE $[\![q_w[j \cdot |\Omega_{a_2}| + k]]\!] \leftarrow \textsc{Combine}([\![v]\!])$ 
               \ENDFOR
           \ENDFOR
       \ENDIF
   \ENDFOR
   \STATE \textbf{return} $[\![q]\!]$, $\forall w \in \mathcal{W}$
   \end{algorithmic}
\end{myprotocol}

A primary challenge in this step is determining an efficient encrypted memory layout for the sparse OHE data. 
A naive design that packs the entire $N \times \sum \omega_i$ matrix into a single ciphertext is infeasible for three key reasons: 
(i) \textit{Capacity constraints}, as the total data size typically exceeds the slot limits of standard polynomial parameters; 
(ii) \textit{Alignment costs}, as computing 2-way marginals would necessitate expensive homomorphic rotations to align columns corresponding to different attributes; and 
(iii) \textit{Ragged domains}, as varying attribute domain sizes $|\Omega_w|$ make uniform packing inefficient.

To overcome these limitations, we propose a \textit{column-wise SIMD packing} strategy. 
Each binary column of the OHE matrix is encrypted into a distinct ciphertext, utilizing SIMD slots to process up to $L(>N)$ records simultaneously. 
This design ensures that columns are naturally pre-aligned for interaction, effectively reducing the computational complexity by a factor of $L$. 
Consequently, computing a 1-way marginal reduces to summing the encrypted column for each bin (requiring $O(\log L)$ rotations at multiplicative depth 0). Similarly, 2-way marginals for attribute pair over $a_1, a_2$, with domain size $\omega_{a_1}, \omega_{a_2}$, require only element-wise multiplication of the respective ciphertexts followed by summation (multiplicative depth 1), eliminating the need for rotation operations during the multiplication phase. 
See App. \ref{app:fhe} for illustrative explanations.
From an FHE perspective, this approach is highly scalable: the multiplicative depth depends only on the marginal degree ($k-1$), not on the dataset size $N$ or domain size $\Omega$. 
The resulting per-bin scalar ciphertexts are subsequently consolidated into a single packed ciphertext using the \textsc{Combine} sub-protocol (see Protocol \ref{prot:combine} in App \ref{prot:combine_desc}), which employs a selector plaintext to arrange each scalar into its designated slot for the downstream Select step.

\noindent
\textbf{FHE Protocol for SELECT ($\pSEL$).}
The select step (Protocol \ref{prot:select}) identifies the workload query $w^*$ that maximizes the error between the real and synthetic data distributions. 
For each query $w \in \mathcal{W}$, the CE computes an encrypted error vector $[\![d_j]\!]$ by subtracting the estimated answers $[\![q_w(\hat{D})]\!]$ from the encrypted true answers $[\![q_w(D)]\!]$. 
The core challenge in this step lies in defining a numerically stable quality score $s(w, D)$ within the encrypted domain.

Standard AIM implementations rely on the $L_1$ norm, but implementing this in FHE presents critical hazards due to the need to approximate the non-polynomial absolute value function. Polynomial approximations face a strict trade-off: low-degree versions degrade utility, while high-accuracy solutions (e.g., the degree-1024 model from the FHERMA challenge\footnote{\url{https://fherma.io/content/65de3f45bfa5f4ea4471701c}}) are too computationally expensive for iterative use. Furthermore, these approximations are only stable within fixed ranges (typically $[-1, 1]$), creating a risk of divergence or decryption failure when processing dynamic, data-dependent error magnitudes that exceed these bounds.

\begin{myprotocol}
   \caption{$\pSEL$: FHE Protocol for \textbf{SELECT}}
   \label{prot:select}
    \textbf{Input:} Candidate workload queries $Q_C \subseteq \mathcal{W}$, 
    encrypted marginals $[\![q]\!]$, 
    estimated marginals $\hat{q}$,
    privacy parameter $\varepsilon$, 
    sensitivity $\Delta$,
    bias vector $\mathbf{b}$, weights vector $\mathbf{w}$ \\
    \textbf{Output:} Selected query $w^*$
   \begin{algorithmic}[1]
   \FORALL{$w \in Q_C$}
       \STATE $[\![\hat{q}_w]\!] \leftarrow \textsc{Enc}(\hat{q}_w)$ \COMMENT{encrypt estimates}
       
       \FOR{$j \leftarrow 1$ to $\omega_w$}
           \STATE $[\![d_j]\!] \leftarrow [\![q_w[j]]\!] - [\![\hat{q}_w[j]]\!]$
       \ENDFOR
       \STATE $[\![\mathbf{d}_w]\!] \leftarrow \textsc{Combine}([\![d_1]\!], \ldots, [\![d_{\omega_w}]\!])$ 
       \STATE $[\![s_w]\!] \leftarrow \textsc{SquaredNorm}([\![\mathbf{d}_w]\!])$ 
       
       \STATE $[\![s_w]\!] \leftarrow ([\![s_w]\!] - b_w) \cdot w_w$
       
       \STATE \textit{// Add Gumbel Noise (Encrypted Gumbel-Max)}
       \STATE $[\![s_w]\!] \leftarrow [\![s_w]\!] + [\![G_w]\!] \cdot \frac{2\Delta}{\epsilon}$
       
       \STATE $s_w \leftarrow \textsc{Dec}([\![s_w]\!])$
   \ENDFOR
   \STATE $w^* \leftarrow \arg\max_{w \in Q_C} \; s_w$
   \STATE \textbf{return} $w^*$
   \end{algorithmic}
\end{myprotocol}

To ensure operational robustness, we propose to use the \textit{squared $L_2$ norm} instead. 
This formulation replaces the approximate absolute value with homomorphic squaring -- a native CKKS operation consuming only a single level of multiplicative depth. 
This eliminates both the accuracy-efficiency tradeoff and the range sensitivity issues, providing a deterministic protocol that is unconditionally stable for arbitrary error magnitudes. 
We propose the corresponding quality score as below : 
\begin{equation}
s^t(w, D) =  \alpha_w \left(\left\| q_w(D) - q_w(\hat{D}_{t-1}) \right\|_2^2 - \sigma_t^2 \omega_w\right)
\label{eqn:l2_score_sel}
\end{equation}
Here, $\alpha_w = \sum_{x \in \mathcal{W}} c_x \cdot |w \cap x|$, where $c_x$ is the weight assigned to the workload query $x$. The penalty term $\rho = \sigma^2\omega_w$ (see Theorem \ref{thm:bias}) and
The sensitivity $\Delta  s(w, D) = \max_{w \in \mathcal{W}} \alpha_w(2N + 1)$ (see Theorem \ref{thm:sens}).


To compute the $w^*$ (see select step in Sec.~\ref{sec:sdgpremlim}), the CE scales the pre-encrypted Gumbel noise  by $\frac{2 \cdot \Delta}{\varepsilon}$ where $\Delta = \max_{w \in \mathcal{W}} \Delta(s(w, \cdot))$, and adds the scaled noise directly to the encrypted quality scores. 
The resulting noisy scores are sent to the CSE for decryption. 
Because the DP noise is applied within the encrypted domain, the CSE can safely decrypt the scores and compute the $\arg\max$ in the clear. This output is permissible under the post-processing property of DP, shifting the complex selection logic out of the FHE circuit while maintaining strict input privacy.

Empirical results confirm that while the $L_1$ implementation suffers from utility loss due to approximation errors, the squared $L_2$ approach produces results nearly identical to the plaintext baseline (see Sec.~\ref{sec:results}, Table \ref{tab:utility}).

\noindent
\textbf{FHE Protocol 
($\pMSR$).
} 
The measure step (See Protocol \ref{prot:measure_he} and details in the App. \ref{app:fhe}) computes a noisy marginal for a given workload $w^*$ by adding encrypted Gaussian noise to the encrypted marginal answer, implementing the Gaussian mechanism entirely within the encrypted domain. 




%
%
\section{Results}\label{sec:results}


\begin{figure*}[t!] 
    \centering
    \begin{subfigure}[b]{0.48\linewidth}
        \centering
        \includegraphics[width=0.7\linewidth]{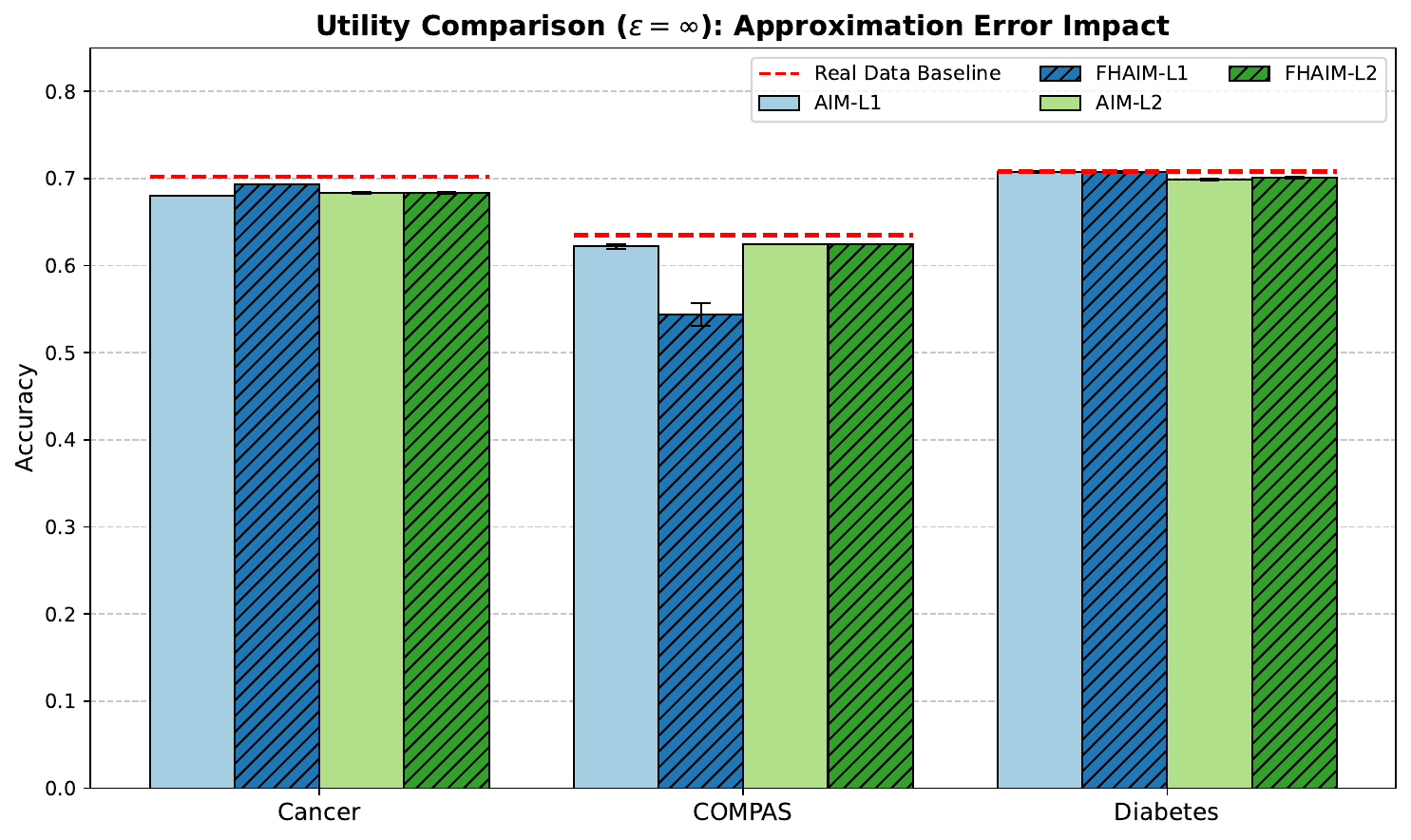}
        \caption{\textbf{$\epsilon=\infty$ (No DP Noise).} $\fhaim$-L1 suffers accuracy loss (e.g., COMPAS) purely due to approximation errors, while $\fhaim$-L2 remains stable.}
        \label{fig:eps_inf}
    \end{subfigure}
    \hfill 
    \begin{subfigure}[b]{0.48\linewidth}
        \centering
        \includegraphics[width=0.7\linewidth]{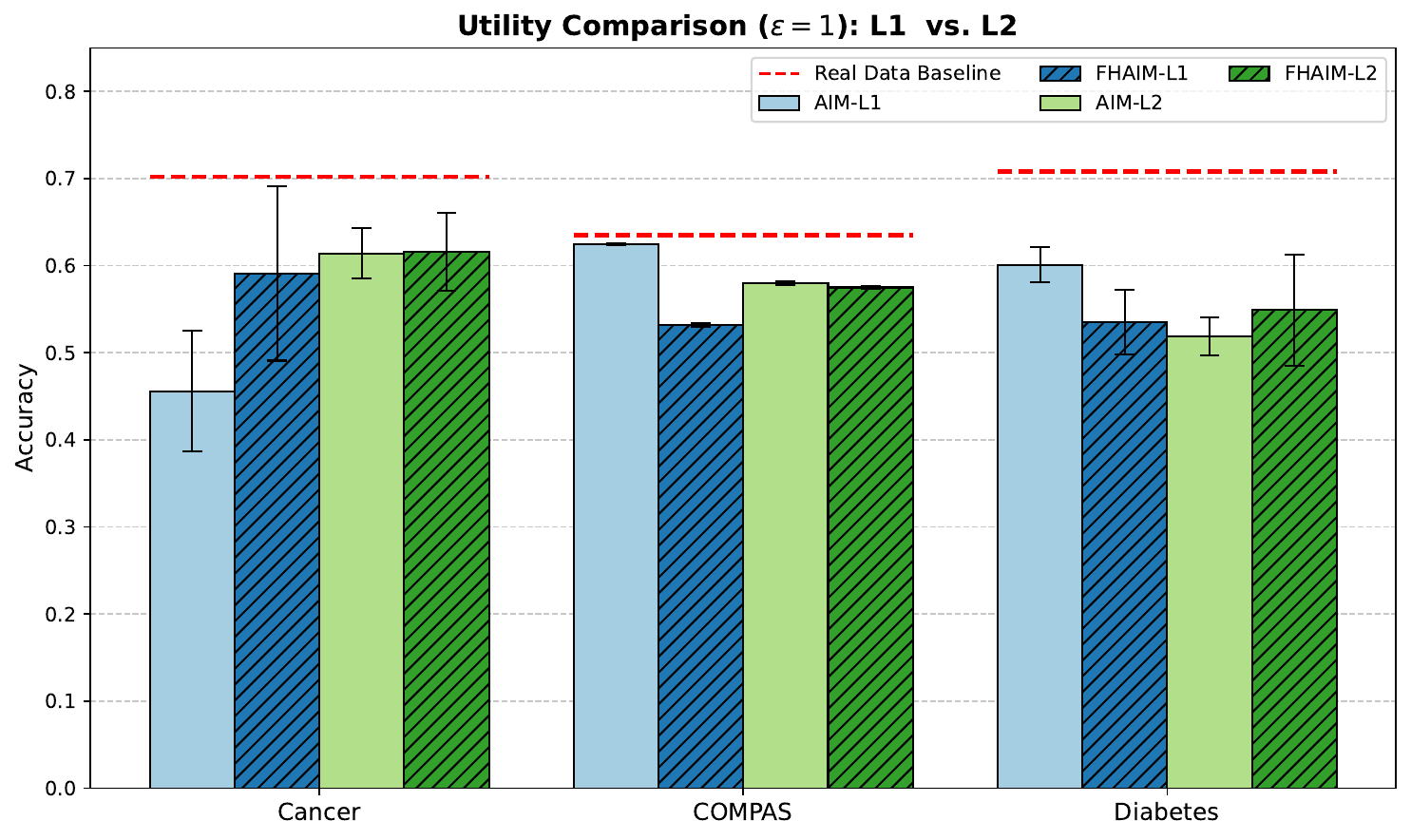}
        \caption{\textbf{$\epsilon=1$ (High Privacy).} $\fhaim$-L2 (hatched green) matches the AIM-L2 baseline, whereas $\fhaim$-L1 degrades significantly due to instability.}
        \label{fig:eps_1}
    \end{subfigure}
    
    \caption{\textbf{Utility Comparison.} Classification accuracy across three datasets. The results isolate the impact of polynomial approximation errors in the $L_1$ norm (blue bars) versus the stability of the exact squared $L_2$ norm (green bars). Red dashed lines indicate real data performance.}
    \label{fig:utility_combined}
\end{figure*}

We evaluate $\fhaim$ on three datasets: breast-cancer \cite{misc_breast_cancer_14}, prison recidivism (COMPAS)\footnote{\url{https://www.propublica.org/datastore/dataset/compas-recidivism-risk-score-data-and-analysis}} \cite{compas}, and diabetes \cite{smith1988using}. Each dataset is randomly split into train and test subsets using a 80\%/20\% split ratio. 
More details on the datasets are in Appendix \ref{app:data}.
%

\subsection{Utility}\label{sec:utility}
In 
Figure \ref{fig:eps_1} and Figure \ref{fig:eps_inf}
, we report the statistical utility of the generated data as the workload error ($\Delta)$:
$
    \Delta(D,\widehat{D}) = \frac{1}{|\mathcal{W}|} \sum_{w \in \mathcal{W}} \| q_{w}(D) -  q_{w}(\widehat{D}) \| 
$.
We also report the utility of the generated data by training a logistic regression model on the synthetic data and evaluating the trained models on the test data. 
Figures \ref{fig:eps_inf} and \ref{fig:eps_1} compare three experimental settings: the real data baseline (red dashed lines), the non-private $\epsilon=\infty$ setting (Fig. \ref{fig:eps_inf}), and the strict $\epsilon=1$ privacy setting (Fig. \ref{fig:eps_1}) . The results validate the privacy-utility trade-off: model fidelity is highest on real data and exhibits a natural, expected reduction under strict privacy constraints ($\epsilon=1$), reflecting the unavoidable cost of rigorous differential privacy. The $\epsilon=\infty$ results (Fig. \ref{fig:eps_inf}) confirm that without noise, the underlying synthesis logic is sound, while the $\epsilon=1$ case (Fig. \ref{fig:eps_1}) demonstrates that $\fhaim$-L2 maintains this stability even under high-privacy conditions.

The $\varepsilon=\infty$ results (Figure \ref{fig:eps_inf}) isolate structural fidelity, revealing a significant utility gap where $\fhaim$-L1 underperforms the plaintext baseline due to accumulated errors from approximating the non-polynomial $L_1$ norm . Conversely, $\fhaim$-L2 eliminates this gap by leveraging the exact homomorphic squaring of the $L_2$ norm, achieving parity with the plaintext model . This robustness persists under strict privacy constraints ($\varepsilon=1$, Figure \ref{fig:eps_1}), confirming that $\fhaim$-L2 preserves utility while providing input privacy, with only minor stochastic variations attributable to DP noise.

\subsection{Computation Cost}\label{sec:efficiency}

\label{sec:comp_cost}

We evaluate the efficiency of $\fhaim$ by analyzing the runtime and memory footprint across the three phases of the protocol: \textsc{Compute}, \textsc{Select}, and \textsc{Measure} on a 16-Core CPU and 64 GB RAM running Ubuntu 24.04 LTS. Experiments were conducted on a standard cloud instance. Figure \ref{fig:exec-time} illustrates the runtime breakdown, and Table \ref{tab:exec-memory} details peak memory usage.

\noindent\textbf{Runtime Analysis: The Amortization of Privacy.}
As illustrated in Figure \ref{fig:exec-time}, the computational cost is heavily skewed toward the one-time initialization phase, allowing the iterative steps to proceed efficiently.
\begin{itemize}[noitemsep, topsep=0pt,parsep=0pt,partopsep=0pt, leftmargin=*]
    \item \textbf{Initialization (Compute Step):} The \textsc{Compute} phase is the most computationally intensive, ranging from approximately 5 minutes (COMPAS) to 12 minutes (Cancer). Crucially, this is a \textit{one-time cost}. Because the generation of 1- and 2-way marginals depends only on the raw data and the workload $\mathcal{W}$, it is executed once prior to the iterative training loop. As expected, the runtime is identical for both L1 and L2 variants, as the choice of error metric does not influence the initial marginal computation.
    
    \item \textbf{Iterative Training (Select \& Measure):} The training loop is dominated by the \textsc{Select} step. The $\fhaim$-L2 protocol incurs a modest computational overhead compared to $\fhaim$-L1 (e.g., $\sim 15\%$ increase for Cancer). While the Squared $L_2$ norm requires fewer multiplicative levels (depth 1) than the polynomial approximation of $L_1$ (depth 4 for degree-10), the operations involve larger ciphertext scales to preserve precision without rescaling, leading to slightly longer execution times. However, this trade-off is negligible in practice: an average iteration takes $\sim 1-2$ minutes. For a typical run of $T=16d$ iterations, the total training time remains within the manageable range of 15 to 45 minutes for all datasets tested.
    
    \item \textbf{The Measure Step:} The \textsc{Measure} step is extremely efficient, executing in under 2.5 seconds per iteration across all datasets. Its contribution to the total runtime is effectively negligible.
\end{itemize}

\noindent\textbf{Memory Footprint.}
Table \ref{tab:exec-memory} reports the peak memory usage during execution. $\fhaim$ demonstrates a remarkably low memory footprint, with peak usage remaining under \textbf{32 MB} for all datasets.
While $\fhaim$-L2 generally consumes more memory than $\fhaim$-L1 (e.g., 31.72 MB vs. 11.12 MB for COMPAS), this increase is an artifact of the ciphertext maintenance required for the squared norm operations. Nevertheless, these requirements are orders of magnitude lower than the available RAM on standard commercial hardware, confirming that $\fhaim$ can be deployed on lightweight computing instances without memory bottlenecks.

\begin{table}[!ht]
\small
\centering
\caption{Peak Memory Usage (MB)}
\label{tab:exec-memory}
\footnotesize{
\begin{tabular}{lcc}
\toprule
Dataset & $\fhaim$-L1 & $\fhaim$-L2 \\
\midrule
Breast Cancer & 18.95 & 26.28 \\
COMPAS & 11.12 & 31.72 \\
Diabetes & 13.72 & 26.57 \\
\bottomrule
\end{tabular}
}
\end{table}

\noindent\textbf{Complexity and Scalability.}
The use of SIMD packing ensures that our runtime scales efficiently with dataset size $N$. As discussed in Section \ref{sec:prot}, the multiplicative depth of our circuit depends only on the marginal degree $k$, not $N$. Consequently, increasing the number of records primarily affects the number of ciphertexts processed in parallel, which scales linearly $O(N/L)$ rather than super-linearly. This confirms that $\fhaim$ is scalable to larger datasets without requiring exponential increases in compute resources.

\begin{figure}[!ht]
    \centering
    \includegraphics[width=\linewidth]{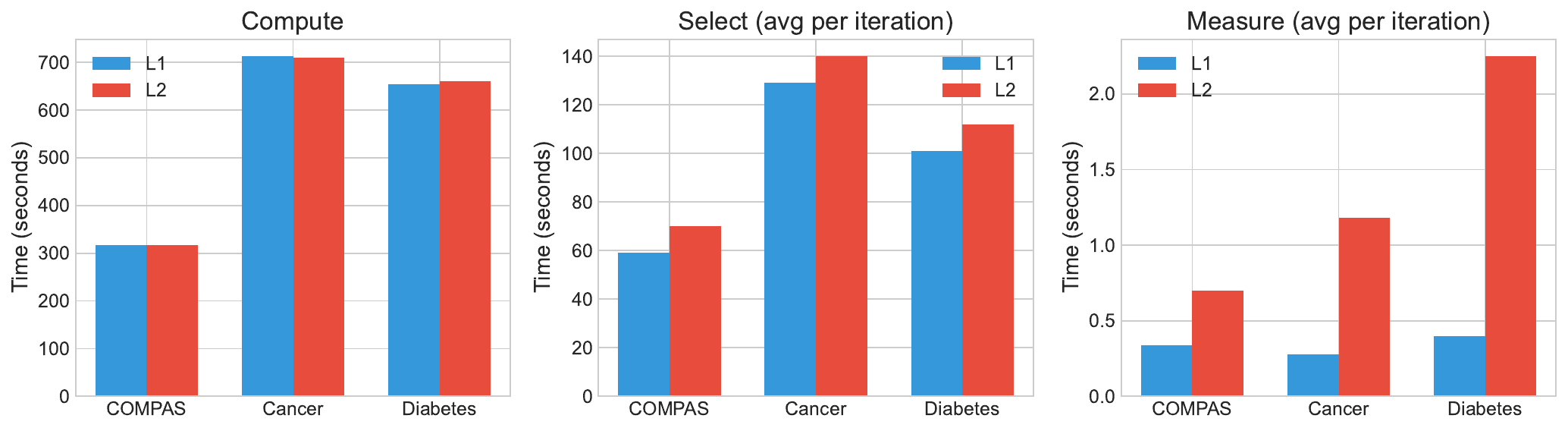}
    \caption{\textbf{Runtime Analysis.} Breakdown of execution times (in seconds) for the Compute, Select, and Measure protocols across datasets. The one-time \textbf{Compute} phase (left) dominates the total runtime and is invariant to the choice of norm. In the iterative \textbf{Select} (center) and \textbf{Measure} (right) phases, the exact Squared $L_2$ protocol (red) incurs a modest computational overhead compared to the degree-10 $L_1$ approximation (blue), reflecting the trade-off between the raw speed of low-degree approximations and the numerical stability of the exact squared norm.}
    \label{fig:exec-time}
\end{figure}



%
%
\section{Conclusion}\label{sec:conc}
We introduced $\fhaim$, the first framework enabling input‑private synthetic data generation by training tabular SDG models directly over fully homomorphically encrypted data. By combining novel FHE protocols for marginal computation, DP noise injection, and encrypted query selection, $\fhaim$ demonstrates that both input privacy and strong output privacy can be achieved within an efficient SDG‑as‑a‑service workflow. Our experiments show that $\fhaim$ preserves the statistical utility and downstream ML performance of the marginal-based AIM SDG algorithm, especially under the more stable squared‑$L_2$ objective, while maintaining practical runtimes on real datasets. These results establish FHE‑based SDG as a viable and scalable primitive for privacy‑preserving data sharing, and open new avenues for secure outsourcing of machine‑learning workflows where both data confidentiality and formal privacy guarantees are essential. 

\section{Acknowledgements}
This material is based upon work supported by the National Science Foundation under Grant Nos.~2451163, 2523406, CCF-2523407, and CNS-2413232, and by NSF NAIRR 240485 (Cloudbank AWS) and NSF NAIRR 240091 (TACC Frontera). This research was, in part, funded by the National Institutes of Health (NIH) Agreement No. 1OT2OD032581. The views and conclusions contained in this document are those of the authors and should not be interpreted as representing the official policies, either expressed or implied, of the NIH.

\section*{Impact Statement}
This work aims to advance privacy-preserving machine learning by enabling synthetic data generation directly on encrypted data. The primary societal benefit is facilitating the safe sharing of sensitive information in critical domains -- such as healthcare and finance -- without requiring data holders to trust a central server. We acknowledge that the use of Fully Homomorphic Encryption introduces a significant computational overhead compared to plaintext approaches, resulting in a higher energy footprint. We do not foresee other immediate negative societal consequences beyond these resource considerations.

\bibliography{ref}
\bibliographystyle{icml2025}

\newpage
\appendix
\onecolumn

\section{FHE primitives}\label{app:fhe}
Homomorphic encryption is a cryptographic method that enables data processing while the information remains in an encrypted state. Because the raw data is never exposed during computation, this technique is essential for secure cloud computing and third-party data processing.

\noindent
Formally, the scheme consists of a tuple of four probabilistic polynomial-time algorithms $\Pi_{\text{HE}} = (\text{KeyGen}, \text{Enc}, \text{Eval}, \text{Dec})$ defined as follows:

\begin{itemize}
    \item $\text{HE.KeyGen}(\lambda)$: 
    Given the security parameter $\lambda$ (which determines the polynomial degree $N$ and modulus $Q$), the key generation algorithm outputs a key pair $(pk, sk)$ and the related cryptographic context. 
    Crucially for our $\fhaim$ implementation, this context includes auxiliary \textit{evaluation keys} (such as relinearization and rotation keys) required to perform homomorphic multiplication and SIMD slot rotations efficiently.

    \item $\text{HE.Enc}(pk, m)$: 
    The encryption algorithm takes in the public key $pk$ and a plaintext message $m$, then outputs the ciphertext $c$. 
    In the CKKS scheme used here, $m$ is typically a vector of real numbers encoded into the slots of a polynomial ring. 
    For our purposes, $m$ corresponds to the flattened vectors of one-hot encoded attributes or pre-generated noise samples.

    \item $\text{HE.Eval}(c, f)$: 
    The encrypted evaluation algorithm takes in a ciphertext message $c$ and a function $f$ (represented as an arithmetic circuit), then outputs the computation result $c'$. 
    Here, $f$ encompasses the arithmetic operations required for marginal computation (additions and multiplications) and the squared $L_2$ error metric evaluation.

    \item $\text{HE.Dec}(sk, c')$: 
    The decryption algorithm takes in the secret key $sk$ and a ciphertext message $c'$, then outputs the plaintext $m'$. 
    Due to the approximate arithmetic of CKKS, the recovered message satisfies $m' \approx f(m)$, containing a negligible error term provided the accumulated noise remains within the decryption radius.

    \item $\text{HE.Add}(c_1, c_2)$: 
    The homomorphic addition algorithm takes two ciphertexts $c_1$ and $c_2$ (encrypting $m_1$ and $m_2$) and outputs a ciphertext $c_{\text{add}}$ such that $\text{HE.Dec}(sk, c_{\text{add}}) \approx m_1 + m_2$. 
    This operation is computationally inexpensive and consumes negligible multiplicative depth. 
    In Protocol \ref{prot:comp_he}, we utilize this extensively to aggregate one-hot encoded counts for 1-way marginals.

    \item $\text{HE.Mult}(c_1, c_2)$: 
    The homomorphic multiplication algorithm outputs a ciphertext $c_{\text{mult}}$ encrypting the element-wise product $m_1 \cdot m_2$. 
    This operation significantly increases the noise level and consumes one level of multiplicative depth. 
    It is typically followed by a \textit{relinearization} step (using evaluation keys) to prevent ciphertext size expansion. 
    We employ this primitive for computing 2-way marginal interactions and for the squaring operation in the $L_2$ quality score calculation.

    \item $\text{HE.Rot}(c, k)$: 
    Given a ciphertext $c$ encrypting a vector $m$, the rotation algorithm outputs a ciphertext $c_{\text{rot}}$ encrypting a cyclically shifted vector $m'$ where $m'[i] = m[(i+k) \pmod{L}]$. 
    This operation requires specific rotation keys generated during $\text{HE.KeyGen}$. 
    In Protocol \ref{prot:select}, this primitive is essential for the \textsc{Combine} sub-protocol, allowing us to move scalar values into specific SIMD slots to pack error vectors efficiently.
\end{itemize}

\paragraph{$\pMSR$: FHE Protocol for \textbf{MEASURE} }
For each bin $j \in \{1,2,\ldots,|\Omega_w|\}$, the CE scales a pre-encrypted unit Gaussian noise sample $[\![z_j]\!]$ by $\sigma$ and adds it to the encrypted marginal entry $[\![q_w[j]\!]$. The resulting noised ciphertext is then sent to the CSE for decryption. Since the DP noise is added before decryption, the CSE only ever observes the already-privatized marginal $\hat{q}_{w^*}$, ensuring that the true marginal is never revealed in the clear to any entity.

\begin{myprotocol}
   \caption{$\pMSR$: FHE Protocol for \textbf{MEASURE}}
   \label{prot:measure_he}
    \textbf{Input:} Selected query $w^*$, 
    encrypted marginal $[\![q_{w^*}]\!]$, 
    noise scale $\sigma$, 
    encrypted Gaussian noise buffer $[\![\mathbf{z}]\!]$, 
    counter $c$ \\
    \textbf{Output:} Noised marginal $\tilde{q}_{w^*}$
   \begin{algorithmic}[1]
   \STATE \textit{//number of attributes in $w$}
   \STATE $m \leftarrow |{w^*}|$
   \FOR{$j \leftarrow 1$ to $\omega_{w^*}$}
       \STATE $[\![\tilde{q}_{w^*}[j]]\!] \leftarrow [\![q_{w^*}[j]]\!] + \sigma \cdot [\![z[c + j]]\!]$
       
       \STATE $\tilde{q}_{w^*}[j] \leftarrow \textsc{Dec}([\![\tilde{q}_{w^*}[j]]\!])$
   \ENDFOR
   \STATE $c \leftarrow c + m$ 
   \STATE \textbf{return} $\tilde{q}_{w^*}$
   \end{algorithmic}
   
\end{myprotocol}

\paragraph{FHE Protocol for Combine ($\pCOMB$).}\label{prot:combine_desc}
The combine sub-protocol (Protocol \ref{prot:combine}) packs $k$ scalar ciphertexts — each holding a value in slot 0 — into a single ciphertext vector where position $i$ holds the $i$-th value. This is done using a pre-computed selector plaintext $\mathbf{s}$ with $\mathbf{s}[0]=1$ and all other entries 0. For each scalar ciphertext $[v_i]$, the CE multiplies it by $\mathbf{s}$ to isolate the value in slot 0, then applies a homomorphic rotation by $i$ to shift it into position $i$, and accumulates the result. This sub-protocol is used in both $\pi_{\text{COMP}}$ to pack marginal entries and in $\pi_{\text{SELECT}}$ to pack the error vector for norm computation.
\begin{myprotocol}[!ht]
   \caption{$\pCOMB$:FHE Sub-Protocol for \textbf{COMBINE}}
   \label{prot:combine}
    \textbf{Input:} Scalar ciphertexts $[\![v_1]\!], \ldots, [\![v_k]\!]$, selector plaintext $\mathbf{s}$ with $\mathbf{s}[0] = 1$ and $\mathbf{s}[j] = 0$ for $j > 0$ \\
    \textbf{Output:} Packed ciphertext $[\![\mathbf{v}]\!]$ with $\mathbf{v}[i] = v_i$
   \begin{algorithmic}[1]
   \STATE $[\![\mathbf{v}]\!] \leftarrow \textsc{Enc}(\mathbf{0})$
   \FOR{$i = 0$ to $k-1$}
       \STATE $[\![\mathbf{v}]\!] \leftarrow [\![\mathbf{v}]\!] + \textsc{Rot}(\mathbf{s} \cdot [\![v_i]\!],\; i)$
   \ENDFOR
   \STATE \textbf{return} $[\![\mathbf{v}]\!]$
   \end{algorithmic}
\end{myprotocol}

\paragraph{Note on one-hot encoding (OHE)}
To compute marginals efficiently in FHE, we one-hot encode (OHE) $D$. Each attribute $x_i$ with domain size $\omega_i$ is transformed into a binary vector of length $\omega_i$ where only the index corresponding to the attribute value is set to 1.
In Figure \ref{fig:2-way}, we consider two attributes 'Age' and 'Gender' with $\Omega_{Age} = \{10, 20, 30\}$ and $\Omega_{Gender} = \{M, F\}$. The OHE data then is tranformed to a matrix of size $N \times  (|\Omega_{Age}| +  |\Omega_{Gender}|)$.

\begin{figure}[!ht]
    \centering
    \includegraphics[width=0.7\textwidth]{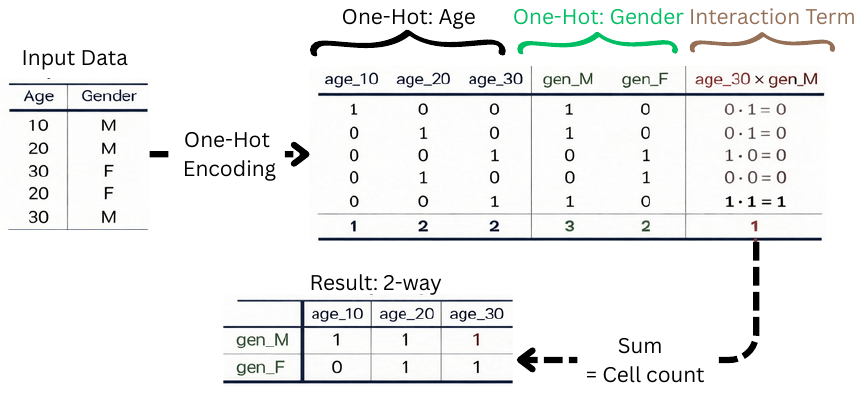}
    \caption{Illustration of 2-way marginal computation}
    \label{fig:2-way}
\end{figure}

To compute 1-way marginal for 'Age' attribute, we need to add each entry in the OHE'd column (as shown in the last row of top right table) and there are 3 such columns. In general, computing the 1-way marginal for attribute $x_i$ with domain $\omega_i$ requires $\omega_i \cdot (N-1)$ additions.

Computing a 2-way marginal is illustrated for computing one cell corresponding to {age\_30, gen\_M} in the resulting 2-way marginal for {'Age', 'Gender'}.  The corresponding to {age\_30, gen\_M} columns are multiplied resulting in a new column shown by the interaction term and then the values are added up. In general, computing the 2-way marginal for attribute pair $(x_i,x_j)$ requires $\omega_i \cdot \omega_j \cdot N$ multiplications to determine the intersection of one-hot encoded features for each record, followed by $\omega_i \cdot \omega_j \cdot (N-1)$ additions to aggregate these counts\footnote{Using comparisons, we would  need $\omega_i \cdot N$ equality checks and $\omega_i \cdot N$ additions for computing a 1-way marginal.} 

 Generalizing this, computing a $k$-way marginal for a subset of attributes $\mathcal{C}$ requires $N(k-1)\prod_{i \in \mathcal{C}} \omega_i$  multiplications and $(N-1)\prod_{i \in \mathcal{C}} \omega_i$ additions.

\section{Noise Generation and DP-in-FHE}\label{app:noise}
A central challenge in DP-in-FHE is generating the randomness required to sample DP noise. We adopt a pre-sampling strategy where DE encrypts and uploads a sufficient collection of unit noise samples during the initial upload phase to support the full AIM execution. The CE then homomorphically scales this noise based on the specific privacy budget $\varepsilon$ and the global sensitivity of the queries. To allow the DE to remain offline and to preserve the client–server deployment model, we derive data-independent upper bounds on the number of required noise samples. The total noise requirement follows directly from the two phases of AIM: (i) the initialization phase, which measures all 1-way marginals (Algorithm 2 in \cite{mckenna2022aim}), and (ii) the iterative select–measure–generate loop (Algorithm 3 in \cite{mckenna2022aim}). Following AIM, we upper bound the number of selection rounds by $T = 16d$.
\begin{itemize}[noitemsep, topsep=0pt]
    \item \textit{Gaussian Noise Samples (Measure)} During initialization, all 1-way marginals are measured once. Furthermore, in each of the $T$ rounds, a single selected marginal is measured. To obtain a data-independent bound on the amount of noise samples required, we assume that each selected marginal has the maximum possible domain size. The total number of required Gaussian noise samples is therefore
    \begin{equation}
    \left( \sum\limits_{i = 1}^d 
    |\omega_i| \right) + \left(16d \cdot \max_{w \in \mathcal{W}} |\omega_w|\right)
    \end{equation}

    \item \textit{Gumbel Noise Samples (Select)}: In each of the $T$ rounds, DP noise is added to the candidate workload $\mathcal{C} \subseteq \mathcal{W}$ during query selection. We conservatively bound the total number of required Gumbel noise samples by
    \begin{equation}
    16d \cdot |\mathcal{W}|
    \end{equation}

\end{itemize}

While this adds a one-time computational burden on the DE, it is a necessary requirement as it is challenging 
for CE to generate secret randomness within the FHE domain. 
Note that the CE cannot add 
its own known plaintext noise, as it then could easily subtract that noise after the CSE decrypts the result, thereby exposing the raw data. By providing encrypted noise, the DE ensures that the statistics are privatized via ``DP-in-FHE'' before they are ever observed in plaintext by the service provider.

The noise samples are chunked and packed into ciphertexts for efficient summation with marginals as needed.

\section{Polynimial approximation of $|x|$}
The L1-norm requires approximating the absolute value function, which has no closed-form polynomial representation. We approximate $|x|$ over $[-1, 1]$ using a degree-10 least-squares polynomial:
\begin{multline}
    p(x) = 0.0556 + 3.6049x^2 - 11.9929x^4 + 24.4175x^6 \\
    \quad - 23.6236x^8 + 8.5577x^{10}
\end{multline}
To extend the approximation to an arbitrary range $[-\alpha, \alpha]$, we evaluate $\alpha \cdot p(x/\alpha)$. Value of $\alpha$ used is 10000.

\section{Optimizations for Select Step}\label{app:select}
We propose the following score function to use squared $L_2$ norm.

\begin{equation}\label{eqn:l2_score_app}
s^t(w, D) =  \alpha_w \left(\left\| q_w(D) - q_w(\hat{D}_{t-1}) \right\|_2^2 - \sigma_t^2 \omega_w\right)
\end{equation}

\begin{theorem}[Sensitivity of Squared $L_2$ Quality Score]\label{thm:sens}
Let $q_w(D)$ be a marginal and $\hat{q}_w$ be the estimated marginal from the graphical model. 
Define the quality score as $s(w,D) = \alpha_w \left(\left\| q_w(D) - q_w(\hat{D}) \right\|_2^2 - \rho\right)$, where  
Under unbounded differential privacy, the global sensitivity of $s(w,.)$ is $\Delta s(w, .) = |\alpha_w|(2N + 1)$, where $N$ is the number of records, $\rho$ is penalty term and $\alpha_w$ is the weight assigned to $w$.
\end{theorem}

\begin{proof}
Let $D$ and $D'$ be neighboring datasets such that $D' = D \cup \{z\}$, where $z$ is a single record. 
Adding $z$ increments exactly one bin in the marginal $q_w(D)$. 
Let $x$ be the count of that bin in $D$, and $\hat{x}$ be the corresponding estimate in $q_w(\hat{D})$. 
The change in the quality score is ($\rho$ cancels out):
\begin{align*}
|s(w,D') - w(w,D)| &= |\alpha_w| \cdot |((x+1) - \hat{x})^2 - (x - \hat{x})^2| \\
&= |\alpha_w| \cdot |(x - \hat{x} + 1)^2 - (x - \hat{x})^2| \\
&= |\alpha_w| \cdot |(x - \hat{x})^2 + 2(x - \hat{x}) + 1 - (x - \hat{x})^2| \\
&= |\alpha_w| \cdot |2(x - \hat{x}) + 1|
\end{align*}
In the worst case, the error $|x - \hat{x}|$ is bounded by the total number of records $N$. 
Thus, the global sensitivity is $\Delta s(w, D) = |\alpha_w|(2N + 1)$. 
\end{proof} 

To apply the exponential mechanism to select a candidate, we use $\Delta  s(w, D) = \max_{w \in \mathcal{W}} \alpha_w(2N + 1)$. The privacy analysis follows zCDP compositions and does not change. The only change here is the sensitivity parameter and not the composition or analysis.

\begin{theorem}[Expected Squared $L_2$ Penalty]\label{thm:bias}
Let $\mathbf{y} = q_w(D) + \mathbf{\eta}$ be a noisy measurement where $\mathbf{\eta} \sim \mathcal{N}(0, \sigma^2 \mathbf{I}_{\omega_w})$. 
The expected squared $L_2$ error introduced solely by the Gaussian noise is $\mathbb{E}[\|\mathbf{\eta}\|_2^2] = \omega_w\sigma^2$ 
\end{theorem}

\begin{proof}

The squared $L_2$ norm of the noise vector $\mathbf{\eta}$ is the sum of the squares of its individual components:
\begin{equation*}
\|\mathbf{\eta}\|_2^2 = \sum_{i=1}^{\omega_w} \eta_i^2
\end{equation*}

By the linearity of expectation, the expected value of a sum is equal to the sum of the expected values:
\begin{equation*}
\mathbb{E}[\|\mathbf{\eta}\|_2^2] = \mathbb{E}\left[\sum_{i=1}^{\omega_w} \eta_i^2\right] = \sum_{i=1}^{\omega_w} \mathbb{E}[\eta_i^2]
\end{equation*}

The expectation of a squared random variable, $\mathbb{E}[X^2]$ is given by:
\begin{align*}
\text{Var}(X) = \mathbb{E}[X^2] - (\mathbb{E}[X])^2 \\
\mathbb{E}[X^2] = \text{Var}(X) + (\mathbb{E}[X])^2
\end{align*}

For Gaussian noise $\eta_i \sim \mathcal{N}(0, \sigma^2)$, the mean ($\mathbb{E}[\eta_i]$) is $0$ and the variance ($\text{Var}(\eta_i)$) is $\sigma^2$. Substituting these values into the identity:
\begin{align*}
\mathbb{E}[\eta_i^2] &= \sigma^2 + (0)^2 \\
\mathbb{E}[\eta_i^2] &= \sigma^2
\end{align*}

Summing over all $\omega_w$ bins, we obtain:
\begin{equation*}
\mathbb{E}[\|\mathbf{\eta}\|_2^2] = \sum_{i=1}^{\omega_w} \sigma^2 = \omega_w\sigma^2
\end{equation*}

This constant $\sigma^2\omega_w$ is the penalty term in the quality score $s(w,D)$ in Equation \ref{eqn:l2_score_sel}.
\end{proof}

\begin{table*}
\centering
\caption{\textbf{Utility of Synthetic Data.} AIM-L1 is the original AIM algorithm}
\label{tab:utility}
\footnotesize
\setlength{\tabcolsep}{3pt}
\resizebox{\linewidth}{!}{%
\begin{tabular}{@{}llccccccccc@{}}
\toprule
& & \multicolumn{3}{c}{\textsc{Cancer}} & \multicolumn{3}{c}{\textsc{Compas}} & \multicolumn{3}{c}{\textsc{Diabetes}} \\
& & \multicolumn{3}{c}{$N=228, d=10, |\Omega|=598{,}752$} & \multicolumn{3}{c}{$N=4120, d=7, |\Omega|=5832$} & \multicolumn{3}{c}{$N=614, d=9, |\Omega|=375{,}000$} \\
\cmidrule(lr){3-5} \cmidrule(lr){6-8} \cmidrule(lr){9-11}
\multicolumn{2}{l}{Method} & $\Delta$ & Acc. & F1 & $\Delta$ & Acc. & F1 & $\Delta$ & Acc. & F1 \\
\midrule
\multicolumn{2}{l}{Real data} & -- & 0.702 & 0.585 & -- & 0.635 & 0.625 & -- & 0.708 & 0.602 \\
\midrule
& AIM-L1 & 0.057 & 0.680 & 0.318 & 0.013 & 0.622 & 0.605 & 0.297 & 0.708 & 0.582 \\
& $\fhaim$-L1 & 0.105 &  0.693 & 0.104 & 0.031 & 0.544 & 0.139 & 0.321 & 0.708 & 0.491 \\
& AIM-L2 & 0.058 & 0.684 & 0.284 & 0.014 & 0.625 & 0.610 & 0.300 & 0.699 & 0.577 \\
\multirow{-4}{*}{\rotatebox[origin=c]{90}{Synthetic}} 
\multirow{-4}{*}{\rotatebox[origin=c]{90}{$\varepsilon=\infty$}} & $\fhaim$-L2 & 0.058 & 0.684 & 0.321 & 0.014 & 0.625 & 0.610 & 0.298 & 0.701 & 0.583 \\
\midrule
& AIM-L1 & 0.415 & 0.456 & 0.338 & 0.019 & 0.625 & 0.628 & 0.361 & 0.601 & 0.299 \\
& $\fhaim$-L1 & 0.475 & 0.591 & 0.349 & 0.035 & 0.532 & 0.504 & 0.382 & 0.535 & 0.355 \\
& AIM-L2 & 0.441 & 0.614 & 0.307 & 0.028 & 0.580 & 0.597 & 0.435 & 0.519 & 0.278 \\
\multirow{-4}{*}{\rotatebox[origin=c]{90}{Synthetic}} 
\multirow{-4}{*}{\rotatebox[origin=c]{90}{$\varepsilon=1$}} 
& $\fhaim$-L2 & 0.476 & 0.616 & 0.495 & 0.028 & 0.575 & 0.620 & 0.444 & 0.549 & 0.218 \\
\bottomrule
\end{tabular}%
}
\end{table*}

\begin{table}[]
    \centering
    \footnotesize{
    \begin{tabular}{l l r r r}
    \toprule
         &  Subprot. & \textsc{Cancer} & \textsc{Compas} & \textsc{Diabetes}\\
    \midrule
        $\pCOMP$ & $\pOM$ & 31s & 25s & 32s \\
                  & $\pTM$ & 681s & 292s & 625s \\
    \midrule
        $\pSEL$  & $\pERR$ & 129/140s & 59/70s & 101/112s \\
                & $\pGUMBEL$ & 4.8s & 3.9s & 4.0s \\

    \midrule
        $\pMSR$ & $\pGAUSS$ & 0.28/1.18s & 0.34/0.70s & 0.40/2.25s \\
    \bottomrule
    \end{tabular}
    }
    \caption{\textbf{Computational Cost}: $\fhaim$ Runtimes (L1/L2)}
    \label{tab:runtimes}
\end{table}

\section{Data, Modelling and Parameters}\label{app:details}
\subsection{Data}\label{app:data}
The breast-cancer dataset has 10 categorical attributes and 285 samples. 
The COMPAS data consists of categorical data. We utilize the same version as 
in \cite{calmon2017optimized}, which consists of 7 categorical features and 7,214 samples.
The diabetes dataset has 9 continuous attributes and 768 samples. 

\subsection{Modelling} To test the utility of the generated data, we train logistic regression models. For breast-cancer, the task is to predict if the cancer will recur. For COMPAS, the task is to predict whether a criminal defendant will re-offend. For the diabetes dataset, the task is to classify a patient as diabetic. We train ML models on real train data and synthetic data independantly. We use a random seed of 42 to train the models. For AIM, we use default parameters, except for $\epsilon$. We note that the original AIM algorithm itself takes too long to run for large datasets (as also reported in \cite{chen2025benchmarking}).




\end{document}